\documentclass[a4paper, 12pt]{article}
\usepackage{graphicx}  
\usepackage{amsmath, amsthm,  amsfonts, amscd,amssymb,latexsym}
\usepackage{url}
\usepackage{setspace}

\newtheorem{theorem}{Theorem}[section]

\newtheorem{lemma}{Lemma}[section]

\title{A comparison of score-based methods for estimating Bayesian networks using the Kullback-Leibler divergence}

\author{
          Jessica Kasza, Patty Solomon\\
\it        School of Mathematical  Sciences \\
\it       University of Adelaide  \\
\it        S.A. 5005, Australia }

\begin{document}

\maketitle{}
\begin{abstract}
\noindent
In this paper, we compare the performance of  two methods for estimating  Bayesian networks from  data containing exogenous variables and random effects. The first method is  fully Bayesian in which a prior distribution is placed on the exogenous variables, whereas the second method, which we call the residual approach, accounts for the effects of exogenous variables by using the notion of  restricted maximum likelihood.  We review the two score-based metrics, then study their performance by measuring the Kullback Leibler divergence, or distance, between the two resulting posterior density functions. The Kullback Leibler divergence provides a natural framework for comparing distributions.  The residual approach is considerably simpler to apply in practice and we demonstrate its utility both theoretically and via simulations.  In particular, in applications where the exogenous variables are not of primary interest,  we show that the potential loss of information about parameters and induced components of correlation, is generally small. 
\end{abstract}

\thispagestyle{myheadings}
\markboth{}{Comparison of Bayesian network scores}

\section{Introduction}

Methods for the estimation of Bayesian networks,  which encode conditional independence relationships of a set of variables, have, until recently, assumed data sets that consist of independent and identically distributed samples, as described in Chapter 16 of \cite{KollerFriedman}. These methods may be split into two categories, called constraint-based and score-based methods, \cite{KollerFriedman, SGS, MMHC}. Recent work by Kasza {\it et al} \cite{KaszaArticle}, has extended the applicability of score-based methods to data sets which do not necessarily consist of independent and identically distributed samples. These authors developed two score metrics  which are extensions of the BGe metric of Geiger and Heckerman,  \cite{LearningGaussianNetworks}, for use in conjunction with score-based methods, to account for complex sampling structures and additional components of variance. The first metric, called the Bayesian score metric, involves placing a prior distribution on the effects of exogenous variables. The second metric, inspired by the notion of restricted maximum likelihood, and called the residual score metric, is non-parametric in the effects of exogenous variables. These two score metrics lead to different posterior distributions for Bayesian network parameters, and a formal comparison of these posterior distributions is necessary to determine if the residual approach  provides a useful alternative to the (fully) Bayesian approach.  This comparison is the subject of the present paper.

In Section \ref{ScoreRevSection}, score-based estimation of Bayesian networks is briefly reviewed, as are the Bayesian and residual score metrics. The posterior distributions obtained using each score metric are also presented here. In Section \ref{InfoLossSection} the posterior distributions are compared using the Kullback Leibler divergence, which in general provides a useful basis for comparing probability density functions. The  comparison of the posterior densities based on the Kullback Leibler divergence  provides justification for the use of the residual score metric in the estimation of Bayesian networks, both theoretically, by simulations and the analysis of data on grape-berry heat-shock genes in  Section \ref{Examples}.

\section{Learning Bayesian networks and estimating parameters}\label{ScoreRevSection}

Bayesian networks were  first introduced by Pearl in  \cite{Pearl}.   A Bayesian network $B = \left( \mathcal{G}, \Theta \right)$, $\Theta =  \left\{ \theta_1, \ldots, \theta_p\right\}$, for a random vector $\boldsymbol{X}=(X_1,   \ldots, X_p)^T$ consists of two components: a directed acyclic graph associated with $\boldsymbol{X}$, $\mathcal{G} = \left(V,E \right)$, with $V = \left\{ X_1, \ldots, X_p\right\}$, $E \subseteq V \times V$, and a set of conditional distributions  $\left\{f(x_i| \boldsymbol{x}_{P_i}, \theta_i)|i = 1,\ldots, p\right\}$. The set $P_i$ consists of those variables $X_j$ such that there is a directed edge from $j$ to $i$ in $\mathcal{G}$: $P_i = \{X_j |  (j,i) \in E\}$. The joint distribution for $\boldsymbol{X}$ may then be written as
\begin{eqnarray}
f\left(\boldsymbol{x} |  \Theta \right) = \prod_{i=1}^p f(x_i| \boldsymbol{x}_{P_i}, \theta_i). \notag
\end{eqnarray}
We make the assumption that $\boldsymbol{X} |  \Theta \sim N(0, \Sigma)$. Bayesian networks are particularly useful as they allow the estimation of covariance matrices for high-dimensional data sets, which contain fewer samples than random variables, since $\Sigma$ can be estimated from the Bayesian network. Additionally, the directed acyclic graph of a Bayesian network encodes information about the conditional dependence relationships between the variables in $\boldsymbol{X}$. The directed Markov properties, as described in Lauritzen \cite{Lauritzen}, for example, allow for more conditional independence relationships to be read directly from the graph $\mathcal{G}$ than could be read from $\Sigma^{-1}$.

Estimation of a Bayesian network for $\boldsymbol{X}$ given a data set $\boldsymbol{d}$ requires estimation of the parameters $\Theta$ and learning the structure of $ \mathcal{G}$.  To learn the structure, score-based methods move through the space of directed acyclic graphs, attempting to find the graph that maximises some score metric. An obvious choice of score metric is the likelihood of a graph, however, the structure that maximises the likelihood is the complete directed acyclic graph, encoding no conditional independence relationships, \cite{KollerFriedman}. Bayesian score metrics such as those considered here avoid this problem of over-fitting. When Bayesian score metrics are used to learn structure, parameters may be estimated using Bayesian techniques.

The Bayesian score of a directed acyclic graph $\mathcal{G}$ for a random variable $\boldsymbol{X}$ is defined to be proportional to the posterior probability of the graph given the data set $\boldsymbol{d}$, \cite{GeigerHeckerman}:
\begin{eqnarray}
S(\mathcal{G}| \boldsymbol{d}) = p(\mathcal{G}) p(\boldsymbol{d}|\mathcal{G}) = p(\mathcal{G}) \int_{\mathcal{P}(p)} p(\boldsymbol{d}|\mathcal{G}, \Theta) p(\Theta | \mathcal{G}) d \Theta,  \label{scoremetric}
\end{eqnarray}
where $p(\mathcal{G})$ is the prior probability of the graph $\mathcal{G}$,  $p(\boldsymbol{d}|\mathcal{G})$ is the marginal likelihood of the data given the structure, and $\mathcal{P}(p)$ is the space of symmetric positive-definite $p\times p$ matrices. We will not consider $p(\mathcal{G})$ any further.

Given the  acyclicity of the graphs considered, $p(\boldsymbol{d}|\mathcal{G}, \Theta) = \prod_{i=1}^p f(\boldsymbol{x}_i | \boldsymbol{x}_{P_i}, \theta_i),$
where $\boldsymbol{x}_i$ is the $n$-vector of samples of $X_i$, and $\boldsymbol{x}_{P_i}$ is the $n \times |P_i|$ matrix of samples of the parents of $X_i$ in the graph $\mathcal{G}$. The usual assumption is that the $n$ samples are independent and identically normally distributed: $$\boldsymbol{x}_i | \boldsymbol{x}_{P_i},\boldsymbol{\gamma}_i, \psi_i \sim N(\boldsymbol{x}_{P_i}\boldsymbol{\gamma}_i, \psi_i I).$$

To get the score metric in Equation \eqref{scoremetric}, prior distributions are required for  $\boldsymbol{\gamma}_i$ and $\psi_i$. As shown by Geiger and Heckerman, \cite{GeigerHeckerman}, in the case of iid samples, to obtain a score metric that scores graphs that encode equivalent sets of independence relationships identically, a property known as score equivalence, the choice of priors for $\boldsymbol{\gamma}_i$ and $\psi_i$ is limited to priors of the form
\begin{eqnarray}
\boldsymbol{\gamma}_i| \psi_i \sim N_{|P_i|} (\boldsymbol{0}, \tau^{-1}\psi_i I ), \quad \psi_i \sim \text{Inverse Gamma}\left( \frac{\delta + |P_i|}{2}, \frac{\tau}{2}\right). \label{OrigPriors} 
 \end{eqnarray}
Given these priors, the BGe score metric of \cite{LearningGaussianNetworks} is obtained, which we denote $S_O(\mathcal{G}| \boldsymbol{d}) = p(\mathcal{G})  \prod_{i=1}^p f_O(\boldsymbol{x}_i | \boldsymbol{x}_{P_i})$. The expression for $f_O$ is provided in Appendix A.
\medskip

\subsection*{Non-independent and identically distributed data}

Often the available data set will be more complex, with non-independent samples, or a complex mean structure including exogenous variables as random effects . Such additional complexities may be accounted for through the inclusion of $m$ exogenous variables in the model,  \cite{Kasza}, \cite{KaszaArticle}. If $Q$ is the $n \times m$ matrix containing data on  $m$ exogenous variables, we assume
$$\boldsymbol{x}_i | \boldsymbol{x}_{P_i},\boldsymbol{\gamma}_i, \psi_i, \boldsymbol{b}_i \sim N(\boldsymbol{x}_{P_i}\boldsymbol{\gamma}_i + Q\boldsymbol{b}_i, \psi_i I),\label{llhd}$$
where the elements in $\boldsymbol{b}_i$ are called the effects of the exogenous variables.

In addition to priors for $\boldsymbol{\gamma}_i$ and $\psi_i$, a prior is required for $\boldsymbol{b}_i$, the effect of the exogenous variables. Kasza {\it et al} \cite{KaszaArticle} note that $\boldsymbol{b}_i$ may be dealt with in two ways, leading to two different score metrics. To satisfy score equivalence, these approaches both use the priors in Equation \eqref{OrigPriors} for $\boldsymbol{\gamma}_i$ and $\psi_i$. 

The first approach, called the Bayesian approach, is to place a prior distribution on $\boldsymbol{b}_i$. An extension of a result in  \cite{GeigerHeckerman} implies that in order for score equivalence to hold, if $var(\boldsymbol{b}_i ) = \Sigma_{\boldsymbol{b}_i}$, $\boldsymbol{b}_i |  \Sigma_{\boldsymbol{b}_i}$ must be normally distributed. We consider prior distributions for  $\boldsymbol{b}_i$ that are of the form $\boldsymbol{b}_i | \psi_i \sim N_m (\boldsymbol{0},  \psi_i V)$, since these are the only priors that result in a score metric with a closed form. The Bayesian score metric is given by
$S_B (\mathcal{G}| \boldsymbol{d}) = p(\mathcal{G})  \prod_{i=1}^p f_V(\boldsymbol{x}_i | \boldsymbol{x}_{P_i})$, where $f_V$ is given in Appendix A.

The second approach, called the residual approach, is non-parametric in $\boldsymbol{b}_i$. This  removes the effects of exogenous variables by using linear combinations of residuals obtained after regressing $\boldsymbol{x}_i$ on the columns of $Q$. This is achieved by pre-multiplying each $\boldsymbol{x}_i$ by $P^T$, where $P$ is an $n \times (n-m)$ matrix such that  $P^TQ = 0$, $P^TP = I$, $PP^T =  I - Q \left( Q^TQ \right)^{-1}Q^T$.
It can then be shown that $P^T\boldsymbol{x}_i \sim N_{n-m} \left( P^T\boldsymbol{x}_{P_i}\boldsymbol{\gamma}_i, \psi_i I \right)$. The residual approach is  related to restricted maximum likelihood estimation, and  is particularly advantageous when the effects of the exogenous variables are included  to improve the estimation of a Bayesian network for $\boldsymbol{X}$, but are  not of  intrinsic interest in themselves. Additionally,  when the prior covariance matrix of $\boldsymbol{b}_i$ cannot be accurately specified, or when the assumption of a normal prior distribution for the  $\boldsymbol{b}_i$ is not warranted, the residual approach is preferable to the Bayesian approach. The residual score metric is given by $S_R(\mathcal{G}| \boldsymbol{d}) = p(\mathcal{G})  \prod_{i=1}^p f_R(\boldsymbol{x}_i | \boldsymbol{x}_{P_i})$, and $f_R$ is shown in Appendix A.

Having used either the Bayesian or residual score metric for learning the graphical structure, parameter estimates may be obtained from posterior distributions. Since posterior estimates of $b_i$ are unavailable from the residual approach, we only consider the posterior distributions of  $\boldsymbol{\gamma}_i$ and $\psi_i$.

Using the priors from the Bayesian approach, the likelihood given in Equation \eqref{llhd} and Bayes' theorem, the following posteriors  are obtained:
\begin{eqnarray}
\boldsymbol{\gamma}_i | \psi_i, \boldsymbol{x}_i, \boldsymbol{x}_{P_i} &\sim& N_{|P_i|} \left( \boldsymbol{\mu}_B , \psi_i \left( \tau I +  \boldsymbol{x}_{P_i}^T H_{V}   \boldsymbol{x}_{P_i} \right)^{-1}\right), \notag \\
 \boldsymbol{\mu}_B &=& \left( \tau I +  \boldsymbol{x}_{P_i}^T H_{V}   \boldsymbol{x}_{P_i} \right)^{-1}\boldsymbol{x}_{P_i} ^T H_{V}  \boldsymbol{x}_{i},\notag \\
\psi_i| \boldsymbol{x}_i, \boldsymbol{x}_{P_i} &\sim& \text{Inverse Gamma} \left(\frac{\delta + n+ |P_i|}{2}, \beta_B \right), \notag \\
\beta_B &=& \frac{\tau}{2} + \frac{1}{2}\boldsymbol{x}_{i}^T H_{V}  \boldsymbol{x}_{i} - \frac{1}{2}\boldsymbol{x}_{i}^T H_{V}  \boldsymbol{x}_{P_i}\left( \tau I +  \boldsymbol{x}_{P_i}^T H_{V}   \boldsymbol{x}_{P_i} \right)^{-1}\boldsymbol{x}_{P_i}^T H_{V}  \boldsymbol{x}_{i}.\notag
\end{eqnarray}
The joint posterior density obtained under the full Bayesian approach is denoted by $f_B(\boldsymbol{\gamma}_i, \psi_i | \boldsymbol{x}_i, \boldsymbol{x}_{P_i})$.

Similarly, using the residual approach, the posteriors can be shown to be
\begin{eqnarray}
\boldsymbol{\gamma}_i | \psi_i, \boldsymbol{x}_i, \boldsymbol{x}_{P_i} &\sim& N_{|P_i|} \left(\boldsymbol{\mu}_R, \psi_i \left( \tau I +  \boldsymbol{x}_{P_i}^T PP^T  \boldsymbol{x}_{P_i} \right)^{-1}\right), \notag \\
\boldsymbol{\mu}_R &= &\left( \tau I +  \boldsymbol{x}_{P_i}^T PP^T  \boldsymbol{x}_{P_i} \right)^{-1}\boldsymbol{x}_{P_i} ^T PP^T\boldsymbol{x}_{i},\notag \\
\psi_i| \boldsymbol{x}_i, \boldsymbol{x}_{P_i} &\sim& \text{Inverse Gamma} \left(\frac{\delta + n-m+ |P_i|}{2}, \beta_R \right), \notag \\
\beta_R &=& \frac{\tau}{2} + \frac{1}{2}\boldsymbol{x}_{i}^T PP^T \boldsymbol{x}_{i} - \frac{1}{2}\boldsymbol{x}_{i}^T PP^T \boldsymbol{x}_{P_i}\left( \tau I +  \boldsymbol{x}_{P_i}^T PP^T \boldsymbol{x}_{P_i} \right)^{-1}\boldsymbol{x}_{P_i}^T PP^T \boldsymbol{x}_{i}. \notag 
\end{eqnarray}
The joint posterior density obtained under the residual approach is denoted by $f_R(\boldsymbol{\gamma}_i, \psi_i | \boldsymbol{x}_i, \boldsymbol{x}_{P_i}) $.

The residual approach does not require the specification of any hyperparameters relating to $\boldsymbol{b}_i$, making it easier to use than the Bayesian approach. Given that in the Bayesian approach, the variance of $\boldsymbol{b}_i$ is dependent upon $\psi_i$, and in turn related to the variance of $\boldsymbol{\gamma}_i$, we may obtain less information about these parameters when the residual approach is used instead of the Bayesian approach. It is important to quantify the difference between the Bayesian and residual approaches in this respect, and this is done in the next section  by measuring the Kullback-Leibler distance between  $f_B(\boldsymbol{\gamma}_i, \psi_i | \boldsymbol{x}_i, \boldsymbol{x}_{P_i})$ and $f_R(\boldsymbol{\gamma}_i, \psi_i | \boldsymbol{x}_i, \boldsymbol{x}_{P_i})$.


\section{Comparison of approaches} \label{InfoLossSection}

Using the Kullback-Leibler divergence as a measure of the distance between the density functions, we show that the distance between the posterior densities for the Bayesian network parameters $\boldsymbol{\gamma}_i$ and $\psi_i$ obtained under the Bayesian and residual approaches is generally small, and decreases as the sample size increases. In this way, theoretical justification for the residual approach is provided.

The Kullback-Leibler divergence, \cite{KLOrig}, between  $f_B(\boldsymbol{\gamma}_i, \psi_i | \boldsymbol{x}_i, \boldsymbol{x}_{P_i})$ and $f_R(\boldsymbol{\gamma}_i, \psi_i | \boldsymbol{x}_i, \boldsymbol{x}_{P_i})$ is given by 
\begin{eqnarray}
D(f_B, f_R) =  \int_{(0, \infty) \times \mathbb{R}^{|P_i|}} \log\left\{ \frac{f_B(\boldsymbol{\gamma}_i, \psi_i | \boldsymbol{x}_i, \boldsymbol{x}_{P_i})}{f_R(\boldsymbol{\gamma}_i, \psi_i | \boldsymbol{x}_i, \boldsymbol{x}_{P_i})} \right\}f_B(\boldsymbol{\gamma}_i, \psi_i | \boldsymbol{x}_i, \boldsymbol{x}_{P_i}) d\boldsymbol{\gamma}_i d\psi_i. \notag
\end{eqnarray}
The exact formula is set out in Appendix B. 
Instead of just considering the divergence associated with $\boldsymbol{\gamma}_i$ and $\psi_i$ associated with a given $X_i$, the divergence associated with $\Sigma$, the covariance matrix of $\boldsymbol{X}$ after marginalising over $\boldsymbol{b}_i$, may be obtained. 
The divergence between $f_B(\Sigma| \boldsymbol{X})$, the posterior density of  $\Sigma$ obtained under the Bayesian approach, and $f_R(\Sigma| \boldsymbol{X})$, the posterior obtained under the residual approach, is then available.

\begin{lemma}\label{Lemma}
If the underlying graphical structure of $\boldsymbol{X}$ is known, the divergence between $f_B(\Sigma| \boldsymbol{X})$ and $f_R(\Sigma| \boldsymbol{X})$ is given by
\begin{eqnarray}
D_{\Sigma} \left\{f_{B}(\Sigma| \boldsymbol{X}), f_{R}(\Sigma| \boldsymbol{X})\right\} = \sum_{i=1}^p D\left\{f_B(\boldsymbol{\gamma}_i, \psi_i | \boldsymbol{x}_i, \boldsymbol{x}_{P_i}),f_R(\boldsymbol{\gamma}_i, \psi_i | \boldsymbol{x}_i, \boldsymbol{x}_{P_i}) \right\}.  \notag
\end{eqnarray}
If the underlying graphical structure of $\boldsymbol{X}$ is not known, bounds for the divergence are given by the divergence for the covariance matrix corresponding to a graph with no edges:
\begin{eqnarray}
D_{\Sigma}^e =  \sum_{i=1}^p D\left\{f_B(\boldsymbol{\gamma}_i, \psi_i | \boldsymbol{x}_i),f_R(\boldsymbol{\gamma}_i, \psi_i | \boldsymbol{x}_i) \right\},  \notag
\end{eqnarray} 
and the divergence for the covariance matrix of an arbitrary full graph:
\begin{eqnarray}
D_{\Sigma}^f =  \sum_{i=1}^p D\left\{f_B(\boldsymbol{\gamma}_i, \psi_i | \boldsymbol{x}_i, \boldsymbol{x}_{1}, \ldots,  \boldsymbol{x}_{i-1}),f_R(\boldsymbol{\gamma}_i, \psi_i | \boldsymbol{x}_i, \boldsymbol{x}_{1}, \ldots,  \boldsymbol{x}_{i-1}) \right\}.  \notag
\end{eqnarray} 
\end{lemma}

\begin{proof}
This result follows directly from the properties of the Kullback Leibler divergence.
\end{proof}

Our main result  is the following Theorem which justifies the use of the residual approach instead of the Bayesian approach:

\begin{theorem}\label{Dinfinity}
As $n\rightarrow \infty$, $D_{\Sigma} \left\{f_{B}(\Sigma| \boldsymbol{X}), f_{R}(\Sigma| \boldsymbol{X})\right\} \rightarrow 0$.
\end{theorem}

\begin{proof}
See Appendix C.
\end{proof}

This Theorem tells us that as sample size increases, the posterior densitites obtained when using the residual metric more closely approximate those obtained using the fully Bayesian approach. Hence, provided the sample size is large enough, the residual approach offers a useful alternative to the fully Bayesian approach.

\section{Examples} \label{Examples}

In this section, the residual and Bayesian approaches are compared using the Kullback-Leibler divergence for some specific data sets. We first consider simulated data sets and then consider a data set consisting of expression levels of grape heat-shock genes.


\subsection{Example 1}

In this example, multiple data sets were simulated from the following system of linear recursive equations:
\begin{eqnarray}
X_{ijk} &=& \sum_{l={1}}^{i-1}\gamma_{i, l}X_{ljk} + b_{ij} + \epsilon_{ijk}, \quad \epsilon_{ijk} \sim N(0, \psi_i), \notag \\
 \boldsymbol{b}_{i}  = ( b_{i1}, b_{i2})^T&\sim& N(0, \psi_i V), \notag \\
\boldsymbol{\gamma}_i = (\gamma_{1,l}, \ldots, \gamma_{i-1,l})^T  &\sim& N(0, \psi_iI),\quad
 \psi_i \sim \text{ Inv. Gamma}(1,2)\notag \\
  i &=& 1,\ldots, 20, \quad j= 1, 2, \quad k = 1, \ldots, n, \notag 
\end{eqnarray}
where the only non-zero $\gamma_{i, l}$s were those corresponding to the edges in the graph of Figure \ref{Eg1graph}, and $V ={\upsilon^{-1}} I$. One hundred data sets were simulated according to this model for each pair $\left( n, \upsilon \right)$, where $n=5,10,20,50,100$ and $\upsilon= 0.001, 0.01, 0.1, 1, 10, 100$. For each of the simulated data sets, $D_{\Sigma}^f$, $D_{\Sigma}^e$, and the divergence corresponding to the true structure were calculated. The key results are summarised in Figure \ref{Eg1Results}. 

 \begin{figure}[t]
\begin{center}
\includegraphics[width=0.5\textwidth]{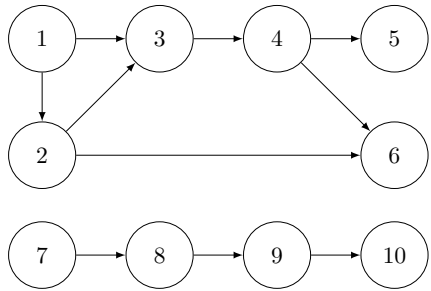}
\caption[Connected components of Example 1.]{Connected components of the underlying graph of Example 1.\label{Eg1graph}}
\end{center}
\end{figure}

As the true graph is quite sparse, the true divergence is closer to that of the empty graph than that of the full graph. For all values of $\upsilon$, as the sample size increases, the divergence decreases, and for all sample sizes, as $\upsilon$ increases, the divergence increases. When $\upsilon$ is large, the samples are \lq\lq similar\rq\rq\: to independent and identically distributed samples, and the fully Bayesian approach allows for this, whilst the residual approach cannot. In these situations, the exogenous variables are over-corrected for when the residual metric is applied.  For larger sample sizes, Figure \ref{Eg1Results}  shows that the divergences obtained for the empty and full graphs provide reasonable approximations to the divergence associated with the true structure.

\begin{figure}[t]
\begin{center}
	\includegraphics[width=1\textwidth]{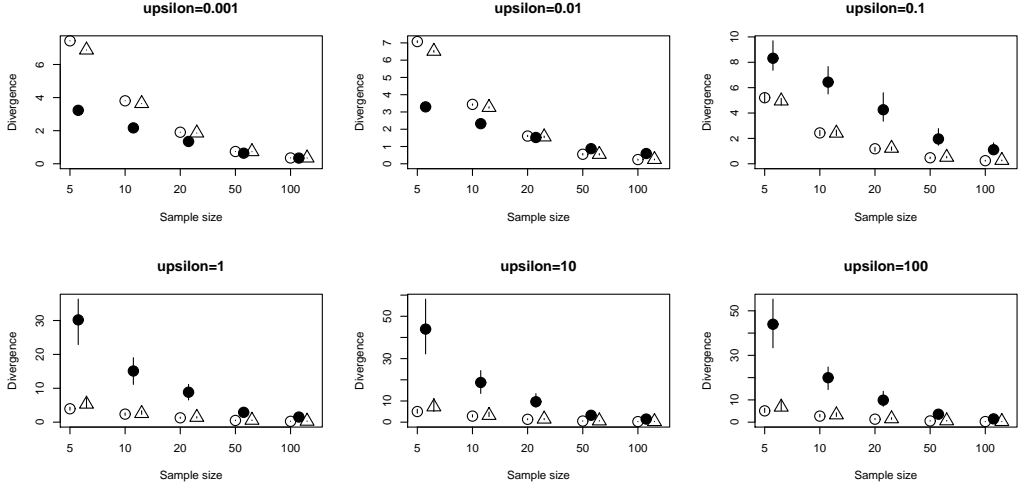}
		\caption[Example 1 results summary.]{The results of Example 1. The open circles represent the median value of $D_{\Sigma}^e$, the filled circles the median  value of $D_{\Sigma}^f$, and the triangles the median loss associated with the true structure, of the 100 simulated data sets for each $(n, \upsilon)$ pair. The vertical bars represent interquartile ranges. Note that the vertical scales of the plots differ.}
\label{Eg1Results}
\end{center}
\end{figure}

These observations are useful in providing guidelines for the use of the residual approach for a given data set. If $\upsilon$ is small, no matter what size the ratio $n/p$ is, the posterior distributions obtained under the Bayesian and residual approaches will be close to each other. In the case where $n/p$ is small, provided $\upsilon$ is large, a similar conclusion is reached. However, for data sets with small values of $n/p$, if the effect of exogenous variables are \emph{a priori} thought to have small variances, the residual approach should be used with caution.

 
\subsection{Example 2}
When the Bayesian approach is used, not much information is  available to guide prior specification of the covariance matrix of the effects, so iid random effects are usually assumed.  In this example, we show that there exist situations where the residual posterior density is closer to the posterior obtained using the data-generating prior, than the posterior density obtained by assuming iid random effects.

Data sets are simulated from the following system of linear recursive equations:
\begin{eqnarray}
X_{ij} &=& \sum_{k={1}}^{i-1}\gamma_{ik}X_{kj} + \sum_{r=1}^3 q_{rj} b_{ir} + \epsilon_{ij}, \quad \epsilon_{ij} \sim N(0, \psi_i), \notag \\
\boldsymbol{b}_{i} &\sim& N_3(0, \psi_i V), \quad \boldsymbol{\gamma}_i = (\gamma_{1,k}, \ldots, \gamma_{i-1,k})^T \sim N(0, \psi_iI),\quad\psi_i \sim \text{ Inv. Gamma}(1,2)\notag \\
  i &=& 1,\ldots, 10, \: \: j= 1, \ldots, 100, \notag 
\end{eqnarray}
where the only non-zero $\gamma_{ik}$s were those corresponding to the edges in the graph of Figure \ref{Eg1graph}, and the $q_{rj}$ are constant across data sets, having been simulated from a standard normal distribution. One hundred data sets were simulated according to this model for each of the following selections for $V$:
$$V_0 = I_3, V_1 =  \begin{pmatrix}10&0&0\\ 0&1&0\\0&0&\frac{1}{10}\end{pmatrix}, V_2 = \begin{pmatrix}  1 & \frac{7}{10} & \frac{6}{10} \\ \frac{7}{10}& 1 & \frac{1}{2} \\ \frac{6}{10} &\frac{1}{2} & 1 \end{pmatrix}, V_3 = \begin{pmatrix}  10 & \frac{7}{10} & \frac{1}{10} \\  \frac{7}{10}& 1 &  \frac{1}{5} \\ \frac{1}{10} &\frac{1}{5} & \frac{1}{10} \end{pmatrix}.$$

Let $f_{\upsilon}$ denote  the posterior distribution obtained under the Bayesian approach when the prior covariance matrix of the effects of exogenous variables is $\upsilon^{-1}I$. For each data set, $D_{\Sigma}(f_B, f_{\upsilon})-D_{\Sigma}(f_B, f_R)$ is calculated for values of $\upsilon$ between 0.0001 and 10, given the true covariance matrix of $\boldsymbol{b}_i$. Figure \ref{DivDiffFig} summarises the median value, and the upper and lower quartiles of $D_{\Sigma}(f_B, f_{\upsilon})-D_{\Sigma}(f_B, f_R)$ for the 100 data sets simulated under each of the four scenarios. The solid lines in Figure \ref{DivDiffFig} correspond to the scenarios where the effects of exogenous variables are heteroscedastic, and the black lines correspond to the scenarios with independent effects.
When $D_{\Sigma}(f_B, f_{\upsilon})-D_{\Sigma}(f_B, f_R)$ is positive, insufficient variation in the data is accounted for by assuming iid effects of exogenous variables with variance $\upsilon^{-1}\psi_i$. As can be seen in  Figure \ref{DivDiffFig}, this happens for all scenarios with increasing probability  as $\upsilon$ increases. Similarly, when $D_{\Sigma}(f_B, f_{\upsilon})-D_{\Sigma}(f_B, f_R)$ is negative, the residual approach removes too much of the variation in the data. 
Given the amount of prior information typically available about the covariance structure of the effects of exogenous variables, this example shows that use of the residual approach will often be preferable to assuming independent and identically distributed effects.

\begin{figure}[t]
\begin{center}
	\includegraphics[width=0.5\textwidth]{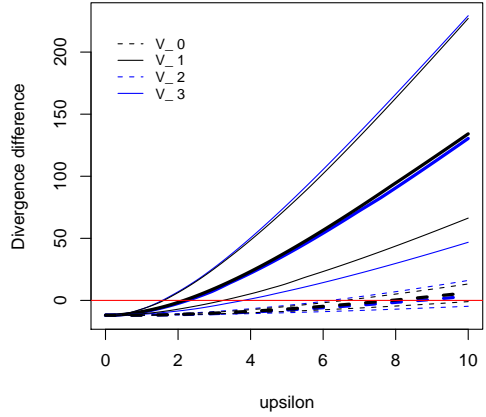}
	\caption[Example 2 results summary.]{Medians (thick lines) and upper and lower quartiles (thin lines) of $D_{\Sigma}(f_B, f_{\upsilon}) - D_{\Sigma}(f_B, f_R)$ for a range of values of $\upsilon$, for the 100 data sets generated for each of the scenarios described in Example 2. Blue lines correspond to dependent effects of exogenous variables, black lines to independent effects. Solid lines correspond to heteroscedastic effects of exogenous variables, dashed lines to homoscedastic effects.}
\label{DivDiffFig}
\end{center}
\end{figure}

\subsection{Grape-berry heat-shock gene example}

We now consider a data set consisting of samples of the expression levels of grape genes, previously discussed in \cite{KaszaArticle}. This data set consists of $n=50$ expression levels of each of $p=26$ grape genes, where the grapes themselves were sampled from three different vineyards located in different wine growing regions of South Australia, Australia.  These $26$ genes are heat-shock genes, see \cite{HSPs}, the expression levels of which are known to be associated with changes in ambient temperature. Accordingly, air temperature at each vineyard was recorded every hour from $5.5$ hours to $0.5$ hours before the grapes were sampled.

The data set considered here is a subset of a larger data set obtained from an Affymetrix chip microarray experiment conducted over the course of three years. Gene expression values were obtained from $174$ grape berry tissue samples: $68$ of these tissue samples were taken from one vineyard, $68$ from the second vineyard, and $38$ from the third. At the first two vineyards, four grape-berry tissue samples were selected each week for $17$ weeks, while at the third, $2$ grape-berry tissue samples were selected each week for $19$ weeks. At each of the vineyards, the first samples were taken at fruit set, when the fertilised grape flowers began to form berries. Samples were then taken each week for a pre-specified number of weeks. In this way, gene expression levels were measured over the course of the development of the grape berries. Of the $174$ samples taken, $162$ had complete temperature records. The data analysed  consist of the samples from each vineyard taken in the third to seventh weeks of sampling, inclusive. The samples from these weeks correspond to a period after fruit set, but before veraison, and it is thought that the relationships between expression levels of genes are relatively stable during this period of berry development, \cite{CoombeGrape, DaviesGrape}.

Let $X_{ij}$ be sample $j$ of gene $i$, $i=1, \ldots, 26$, $j = 1, \ldots, 50$, and let $q_{rj}$ be the data associated with sample $j$ of exogenous variable $r$, where $m$ exogenous variables are included in the model. Then the following model is assumed for each sample of each gene:
\begin{align}
X_{ij} &= \sum_{l\in P_i} \gamma_{il} X_{lj} + \sum_{r=1}^m q_{rj}b_{ir} + \epsilon_{ij}, \: \: \epsilon_{ij} \sim N(0, \psi_i),  \notag \\
 b_{ir} &\sim N\left(0,\frac{\psi_i}{ \upsilon}\right), \quad \gamma_{il} \sim N\left(0, \frac{\psi_i}{\tau}\right),\quad \psi_i \sim \text{Inv. Gamma}\left(\frac{\delta + |P_i|}{2},\frac{\tau}{2}\right).  \label{grapebir}
\end{align}

 For the grape-berry genes  under study here, temperature, which has been  observed directly at the different vineyards, is a known driver of biological activity. Moreover, when two or more genes respond similarly to the same driver of biological activity, the effect is to produce a component of correlation between the corresponding gene expression levels. Thus we should study the effects of temperature as an exogenous variable.  There are also likely to be additional variables which do not correspond directly to a single biological factor such as temperature. For example, the three vineyards are likely to differ in a number of features such as soil type and fertility, moisture and other micro-climate conditions, each of which could potentially influence the expression levels of certain sets of genes. Here the three vineyards are separated by large regional distances, but share the same macro-climate in southern Australia. Thus, vineyards  should be modelled as an additional exogenous variable with potentially considerable heterogeneity.

It is unlikely that the effects of temperature and vineyard are independent and identically distributed. While such a claim may be valid for either the temperature effects or the vineyard effects alone, it is highly unlikely that the effects of temperature and vineyard are identically distributed. There may also be some dependence between the temperature and vineyard effects. Given the difficulty in specifying a joint prior variance matrix of the temperature and vineyard effects, consideration of models including both temperature and vineyard effects simultaneously are unlikely  shed light on the performance of the residual approach to the estimation of Bayesian networks.  We therefore proceed by considering simple models, fitting temperature and vineyards as separate exogenous variables.

\begin{figure}[t]
\begin{center}
	\includegraphics[width=0.8\textwidth]{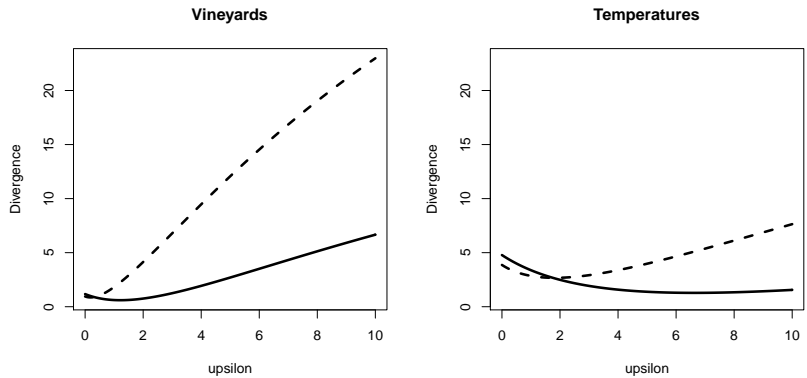}
	\caption[Grape gene example summary.]{Upper and lower bounds of the divergence for the marginal covariance matrix of the 26 grape genes, when vineyards and then temperatures are included as exogenous variables in the analysis. In both graphs, the solid line is the divergence corresponding to the empty graph, and the dashed line is the divergence corresponding to the full graph.}
\label{GrapeDivs}
\end{center}
\end{figure}

Firstly, we consider the vineyards only model, where $m=3$ and in which we are interested in the temperature-induced correlations between genes, and secondly,  the temperature only model, where $m=6$, where we do not remove the components of correlations induced by the vineyard micro-climates. Note that we are ignoring any temperature trend-component in all our models. Although models containing both temperature and vineyard effects may potentially be of interest, the effects may be confounded as explained above, and there is a risk of over-fitting the data. In fact, for the full interaction model fitted to the grape-berry gene data, the effective sample size, $n-m$, would be  zero, and the Kullback Leibler divergence could then be substantially artificially inflated.

Since neither the true network nor the true value of $\upsilon$ is known for this data set, the bounds $D_{\Sigma}^e$ and $D_{\Sigma}^f$ are calculated for a range of values of $\upsilon$. The results are shown in Figure \ref{GrapeDivs}. The left-hand graph in the figure displays the loss of information when the three vineyard effects only are included in the analysis  and the right-hand graph displays the divergence when only the six main temperature effects are included.  
 Figure \ref{GrapeDivs} indicates that for either model and all considered values of $\upsilon$, if the true underlying graph is thought to be sparse, as many biological networks are thought to be, the loss of information about the marginal covariance matrix when the residual approach is used will be minimal. If the true graph of the expression levels of the  genes is thought to be dense, for larger values of $\upsilon$, the figure shows that the divergence for the temperature model will be less than that associated with the vineyard model. The temperature model is naturally likely to be more explanatory, with the higher number of exogenous variables fitted. For either model, the Kullback-Leibler divergence is small and the residual approach metric is of demonstrable practical utility.


\section{Conclusion}

Using the Kullback-Leibler divergence, we have compared two methods for estimating Bayesian networks for data containing exogenous variables and random effects. Provided that the sample size is not too small in a statistical sense, we can conclude that the residual score metric offers a useful alternative to a fully Bayesian approach, with the posterior density functions of key parameters obtained under the two approaches being generally close. Many contemporary bioinformatics  studies are  conducted using substantial sample sizes, often based on many hundreds of samples or patients. Even with smaller studies however,  the results of our simulations and data analysis provide confidence that the residual estimation approach will perform well with small samples in the presence of exogenous variables.


\noindent
{\bf{{Appendix A}}}

$f_O(\boldsymbol{x_i}|\boldsymbol{x}_{P_i})$ is given by the pdf of \\
$$ t_{\delta + |P_i|} \left(\boldsymbol{0}, \frac{\tau}{\delta + |P_i|} \left\{ I - \boldsymbol{x}_{P_i}\left( \tau I + \boldsymbol{x}_{P_i}^T\boldsymbol{x}_{P_i}\right)^{-1}\boldsymbol{x}_{P_i}^T\right\}^{-1} \right). $$

$ f_V(\boldsymbol{x_i}|\boldsymbol{x}_{P_i})$ is given by the pdf of  $$t_{\delta + |P_i|} \left(\boldsymbol{0}, \frac{\tau}{\delta + |P_i|} \left\{ H_{V}  - H_{V} \boldsymbol{x}_{P_i}\left( \tau I + \boldsymbol{x}_{P_i}^TH_{V} \boldsymbol{x}_{P_i}\right)^{-1}\boldsymbol{x}_{P_i}^TH_{V} \right\}^{-1} \right),$$ $$H_{V} = I - Q \left(V^{-1} + Q^TQ \right)^{-1}Q^T.$$

$f_R(P^T\boldsymbol{x_i}|P^T\boldsymbol{x}_{P_i})$ is given by the pdf of 
$$t_{\delta + |P_i|} \left(\boldsymbol{0}, \frac{\tau}{\delta + |P_i|} \left\{ I - P^T\boldsymbol{x}_{P_i}\left( \tau I + \boldsymbol{x}_{P_i}^TPP^T\boldsymbol{x}_{P_i}\right)^{-1}\boldsymbol{x}_{P_i}^TP\right\}^{-1} \right). $$

\noindent
{\bf{{Appendix B}}}

\noindent

The Kullback Leibler divergence between $f_B(\boldsymbol{\gamma}_i, \psi_i | \boldsymbol{x}_i, \boldsymbol{x}_{P_i}) $ and $f_R(\boldsymbol{\gamma}_i, \psi_i | \boldsymbol{x}_i, \boldsymbol{x}_{P_i}) $ is given by

\begin{eqnarray}
D(f_B, f_R) &=& \frac{1}{2} \log\left( \frac{\left| \tau I +  \boldsymbol{x}_{P_i} ^T H_{V}   \boldsymbol{x}_{P_i}  \right|}{\left| \tau I +  \boldsymbol{x}_{P_i} ^T PP^T  \boldsymbol{x}_{P_i}  \right|}\right) + \frac{1}{2}tr\left\{\left( \tau I +  \boldsymbol{x}_{P_i} ^T PP^T  \boldsymbol{x}_{P_i}  \right) \left( \tau I +  \boldsymbol{x}_{P_i} ^T H_{V}   \boldsymbol{x}_{P_i}  \right) ^{-1}\right\} \notag \\
&-&\frac{|P_i|}{2} + \frac{1}{\beta_B} \frac{\delta + n + |P_i|}{4} \left(  \boldsymbol{\mu}_R - \boldsymbol{\mu}_B\right)^T\left( \tau I +  \boldsymbol{x}_{P_i}^T PP^T  \boldsymbol{x}_{P_i} \right)\left(  \boldsymbol{\mu}_R - \boldsymbol{\mu}_B\right) \notag \\
&+& \frac{\delta + n -m+ |P_i|}{2} \log\left(\frac{\beta_B}{\beta_R} \right) + \log\left\{ \frac{\Gamma \left( \frac{\delta + n -m+ |P_i|}{2} \right)}{\Gamma \left( \frac{\delta + n+ |P_i|}{2} \right)} \right\}  \notag \\
&+& \frac{\delta + n + |P_i|}{2}\left( \frac{\beta_R}{\beta_B}-1\right)+ \frac{m}{2} \text{Digamma}\left(\frac{\delta + n + |P_i|}{2} \right). \notag 
\end{eqnarray}

\noindent
{\bf{{Appendix C}}} 

Here we prove Theorem \ref{Dinfinity}. By Lemma \ref{Lemma}, we need only consider the divergence for the parameters for one regression: $D\left\{f_B(\boldsymbol{\gamma}_i, \psi_i | \boldsymbol{x}_i, \boldsymbol{x}_{P_i}),f_R(\boldsymbol{\gamma}_i, \psi_i | \boldsymbol{x}_i, \boldsymbol{x}_{P_i}) \right\}$.

Assume that each $ \boldsymbol{x}_i$ is centred and scaled, so that $ \boldsymbol{x}_i^T \boldsymbol{x}_i = n-1$, and note that $\boldsymbol{x}_{i} ^T H_{V}   \boldsymbol{x}_i < n-1$ and $\boldsymbol{x}_{i} ^T PP^T   \boldsymbol{x}_i < n-1$.

First, consider the log determinant term in $D(f_B, f_R)$:
\begin{eqnarray}
 \frac{1}{2} \log\left( \frac{\left| \tau I +  \boldsymbol{x}_{P_i} ^T H_{V}   \boldsymbol{x}_{P_i}  \right|}{\left| \tau I +  \boldsymbol{x}_{P_i} ^T PP^T  \boldsymbol{x}_{P_i}  \right|}\right)  &=& -\frac{1}{2} \log\left\{\left| \left(\tau I +  \boldsymbol{x}_{P_i} ^T H_{V}   \boldsymbol{x}_{P_i} \right)^{-1} \left( \tau I +  \boldsymbol{x}_{P_i} ^T PP^T  \boldsymbol{x}_{P_i} \right) \right|\right\}\notag \\
&=& -\frac{1}{2}  tr\left[ \log\left\{  \left(\tau I +  \boldsymbol{x}_{P_i} ^T H_{V}   \boldsymbol{x}_{P_i} \right)^{-1} \left( \tau I +  \boldsymbol{x}_{P_i} ^T PP^T  \boldsymbol{x}_{P_i} \right) \right\} \right] \notag \\
&=& -\frac{1}{2}  tr\left( \log\left[  I - \left\{ I- \left(\tau I +  \boldsymbol{x}_{P_i} ^T H_{V}   \boldsymbol{x}_{P_i} \right)^{-1} \left( \tau I +  \boldsymbol{x}_{P_i} ^T PP^T  \boldsymbol{x}_{P_i} \right) \right\}\right] \right),  \notag
\end{eqnarray}
using the Taylor series expansion this can be written as 
$$-\frac{1}{2}  tr\left[ - \sum_{k=1}^{\infty} \frac{1}{k}\left\{ I- \left(\tau I +  \boldsymbol{x}_{P_i} ^T H_{V}   \boldsymbol{x}_{P_i} \right)^{-1} \left( \tau I +  \boldsymbol{x}_{P_i} ^T PP^T  \boldsymbol{x}_{P_i} \right) \right\} ^k\right]. $$
If second- and higher-order terms are ignored, this becomes
\begin{eqnarray}
& &\frac{1}{2}  tr\left\{  I- \left(\tau I +  \boldsymbol{x}_{P_i} ^T H_{V}   \boldsymbol{x}_{P_i} \right)^{-1} \left( \tau I +  \boldsymbol{x}_{P_i} ^T PP^T  \boldsymbol{x}_{P_i} \right) \right\} \notag \\
 &=& \frac{|P_i|}{2} - \frac{1}{2}tr\left\{   \left(\tau I +  \boldsymbol{x}_{P_i} ^T H_{V}   \boldsymbol{x}_{P_i} \right)^{-1} \left( \tau I +  \boldsymbol{x}_{P_i} ^T PP^T  \boldsymbol{x}_{P_i} \right) \right\},  \notag
\end{eqnarray}
terms which cancel with other terms in $D$.

Note also that 
\begin{eqnarray}
\frac{\delta + n -m+ |P_i|}{2} \log\left(\frac{\beta_B}{\beta_R} \right) = \frac{\delta + n -m+ |P_i|}{2} \sum_{k=1}^{\infty} \frac{1}{k}\left( \frac{\beta_B -\beta_R}{\beta_B}\right)^k, \notag
\end{eqnarray}
and since $\frac{\beta_B -\beta_R}{\beta_B} < 1$, second- and higher-order terms in $\sum_{k=1}^{\infty} \frac{1}{k}\left( \frac{\beta_B -\beta_R}{\beta_B}\right)^k$ may be ignored, cancelling with other terms in $D(f_B, f_R)$ so that the divergence becomes
\begin{eqnarray}
D(f_B, f_R) &=& \frac{1}{\beta_B} \frac{\delta + n + |P_i|}{4} \left(  \boldsymbol{\mu}_R - \boldsymbol{\mu}_B\right)^T\left( \tau I +  \boldsymbol{x}_{P_i}^T PP^T  \boldsymbol{x}_{P_i} \right)\left(  \boldsymbol{\mu}_R - \boldsymbol{\mu}_B\right) \notag \\
&+& \log\left\{ \frac{\Gamma \left( \frac{\delta + n -m+ |P_i|}{2} \right)}{\Gamma \left( \frac{\delta + n+ |P_i|}{2} \right)} \right\} + \frac{m}{2} \text{Digamma}\left(\frac{\delta + n + |P_i|}{2} \right). \notag
\end{eqnarray}
Let $n_\ast = \frac{\delta + n + |P_i|}{2}$, and note that as $n$ approaches infinity, so too does $n_\ast$. From \cite{Tricomi}, as $n_\ast \rightarrow \infty$,
\begin{eqnarray}
\frac{\Gamma \left( n_\ast - \frac{ m}{2} \right)}{\Gamma \left( n_\ast \right)} =  \left( n_\ast \right)^{-\frac{m}{2}}\left\{ 1 + \frac{m(m+2)}{8n_\ast} + O\left( \frac{1}{n_\ast^2}\right)\right\}.   \notag
\end{eqnarray}
Hence, for large $n_\ast$,
\begin{eqnarray}
\log\left\{\frac{\Gamma \left( n_\ast - \frac{ m}{2} \right)}{\Gamma \left( n_\ast \right)}\right\} = -\frac{m}{2} \log( n_\ast )+ \log \left\{ 1 + \frac{m(m+2)}{8n_\ast} + O\left( \frac{1}{n_\ast^2}\right)\right\}. \label{lgammaratio}
\end{eqnarray}

From \cite{Abramowitz}, for large values of $n_\ast$
\begin{eqnarray}
 \text{Digamma}\left(n_\ast \right) = \log \left(n_\ast \right) - \frac{1}{2n_\ast} - \frac{1}{12n_\ast^2}  +\frac{1}{120n_\ast^4}  -\frac{1}{252n_\ast^6}  + O\left( \frac{1}{n_\ast^8}\right). \label{digamma}
\end{eqnarray}
Hence, as $n \rightarrow \infty$, $\log\left\{ \frac{\Gamma \left( \frac{\delta + n -m+ |P_i|}{2} \right)}{\Gamma \left( \frac{\delta + n+ |P_i|}{2} \right)} \right\} + \frac{m}{2} \text{Digamma}\left(\frac{\delta + n + |P_i|}{2} \right)$ approaches zero. All that remains to consider is the quadratic term:
\begin{eqnarray}
\frac{1}{\beta_B} \frac{\delta + n + |P_i|}{4} \left(  \boldsymbol{\mu}_R - \boldsymbol{\mu}_B\right)^T\left( \tau I +  \boldsymbol{x}_{P_i}^T PP^T  \boldsymbol{x}_{P_i} \right)\left(  \boldsymbol{\mu}_R - \boldsymbol{\mu}_B\right).  \notag
 \end{eqnarray}
 
Using the following approximations, 
\begin{eqnarray}
\left(\tau + \boldsymbol{x}_{P_i} ^TH_V \boldsymbol{x}_{P_i}\right)^{-1} &\approx&  \text{diag}\left(\frac{1}{\tau +  \boldsymbol{x}_{k} ^T \boldsymbol{x}_{k} - \boldsymbol{x}_{k} ^TQ\left(V^{-1} + Q^TQ \right)^{-1}Q^T\boldsymbol{x}_{k}}\right), \notag \\
\left( \tau + \boldsymbol{x}_{P_i} ^TPP^T \boldsymbol{x}_{P_i}\right)^{-1} &\approx&  \text{diag}\left(\frac{1}{\tau +  \boldsymbol{x}_{k} ^T \boldsymbol{x}_{k}-  \boldsymbol{x}_{k} ^T Q(Q^TQ)^{-1}Q^T\boldsymbol{x}_{k}}\right) \notag 
\end{eqnarray}
where $k \in P_i$, we may write
\begin{eqnarray}
\beta_B \approx\frac{1}{2}\left\{ {\tau} + {\boldsymbol{x}_{i} ^T \boldsymbol{x}_{i}}     -\boldsymbol{x}_{i} ^TQ\left(V^{-1} + Q^TQ \right)^{-1}Q^T \boldsymbol{x}_{i}- \sum_{k \in P_i} \frac{ \left[ {\boldsymbol{x}_{i} ^T \boldsymbol{x}_{k}}     -\boldsymbol{x}_{i} ^TQ\left(V^{-1} + Q^TQ \right)^{-1}Q^T \boldsymbol{x}_{k}\right]^2}{\tau +  \boldsymbol{x}_{k} ^T \boldsymbol{x}_{k} - \boldsymbol{x}_{k} ^TQ\left(V^{-1} + Q^TQ \right)^{-1}Q^T\boldsymbol{x}_{k}} \right\}  \notag
\end{eqnarray} 
and 
\begin{eqnarray}
& &\left(  \boldsymbol{\mu}_R - \boldsymbol{\mu}_B\right)^T\left( \tau I +  \boldsymbol{x}_{P_i}^T PP^T  \boldsymbol{x}_{P_i} \right)\left(  \boldsymbol{\mu}_R - \boldsymbol{\mu}_B\right) \notag \\
&=& \sum_{k \in P_i}\left\{\frac{\left( \boldsymbol{x}_{i}PP^T \boldsymbol{x}_{k}\right)^2}{\tau +\boldsymbol{x}_{k}^T\boldsymbol{x}_{k} - \boldsymbol{x}_{k}^T Q \left( Q^TQ\right)^{-1}Q^T \boldsymbol{x}_{k}} \right\} -2 \sum_{k \in P_i}\left(\frac{\boldsymbol{x}_i ^T H_V \boldsymbol{x}_k \boldsymbol{x}_{i}PP^T \boldsymbol{x}_{k}}{\tau +\boldsymbol{x}_{k}^T\boldsymbol{x}_{k}} \right)\notag \\
&+&\sum_{k \in P_i}\left\{ \frac{(\boldsymbol{x}_i ^TH_V \boldsymbol{x}_k)^2 \left( \tau +\boldsymbol{x}_{k}^T\boldsymbol{x}_{k} - \boldsymbol{x}_{k}^T Q \left( Q^TQ\right)^{-1}Q^T \boldsymbol{x}_{k}\right)} {(\tau +\boldsymbol{x}_{k}^T\boldsymbol{x}_{k} -   \boldsymbol{x}_{k} ^TQ\left(V^{-1} + Q^TQ \right)^{-1}Q^T\boldsymbol{x}_{k})^2}\right\}. \label{LargeUpsQuadTerm}
\end{eqnarray} 
As $n$ increases, $\frac{\delta + n + |P_i|}{\beta_B^{\ast}}$ approaches 1, each of the terms in Equation \eqref{LargeUpsQuadTerm} approaches zero, proving the result.


\begin{thebibliography}{}

\bibitem{Abramowitz}
	Abramowitz, M. and Stegun, I. A., editors. (1970)
	\newblock{\em Handbook of Mathematical Functions with Formulas, Graphs, and Mathematical Tables}.
	\newblock{Washington, D. C.: National Bureau of Standards}.
	
	
\bibitem{CooperHerskovits}
	Cooper, G. F. and Herskovits, E. (1992)
	\newblock{A {B}ayesian method for the induction of probabilistic networks from data}.
	\newblock{\em Machine Learning}, 
	\newblock{\bf 9},
	\newblock{309-347}.	

\bibitem{CoombeGrape}
	Coombe, B. G. (1973)
	\newblock{The regulation of set and development of the grape berry}.
	\newblock{\em Acta Horticulturae}, 
	\newblock{\bf 34},
	\newblock{261-271}.	
	
\bibitem{SAMSI}
	Dobra, A., Hans, C., Jones, B., Nevins, J.R. and West, M.(2004)
	\newblock{Sparse graphical models for exploring gene expression data}.
	\newblock{\em Journal of Multivariate Analysis}, 
	\newblock{\bf 90},
	\newblock{196-212}.	


	
\bibitem{LearningGaussianNetworks}
	Geiger, D. and Heckerman, D. (1994)
	\newblock{Learning Gaussian networks}.
	\newblock{In}
	\newblock{\em Proceedings of the Tenth Conference on Uncertainty in Artificial Intelligence}.


\bibitem{GeigerHeckerman}
	Geiger, D. and Heckerman, D. (2002)
	\newblock{Parameter priors for directed acyclic graphical models and the characterization of several probability distributions}.
	\newblock{\em The Annals of Statistics}, 
	\newblock{\bf 30},
	\newblock{1412-1440}.	
	
\bibitem{Kasza}
	Kasza, J. (2009)
	\newblock{\em {B}ayesian networks for high-dimensional data with complex mean structure}.
	\newblock{Ph. D. thesis, The University of Adelaide}.
	
\bibitem{KaszaArticle}
	Kasza, J. E., Glonek, G. and Solomon, P. (2011)
	\newblock{Estimating {B}ayesian networks for high-dimensional data with complex mean structure and random effects}.
	\newblock{arXiv:1002.2168}.


\bibitem{KollerFriedman}
	Koller, D. and Friedman, N. (2009)
	\newblock{\em Probabilistic Graphical Models: Principles and Techniques}.
	\newblock{The MIT Press}.
	
\bibitem{KLOrig}
	Kullback, S. and Leibler, R. A. (1951)
	\newblock{On information and sufficiency}.
	\newblock{\em The Annals of Mathematical Statistics}, 
	\newblock{\bf 22},
	\newblock{79-86}.	
	
\bibitem{Lauritzen}
	Lauritzen, S. L. (2004)
	\newblock{\em Graphical Models}.
	\newblock{Oxford: Clarendon Press}.

\bibitem{Pearl}
	Pearl, J. (1988)
	\newblock{\em Probabilistic reasoning in intelligent systems}.
	\newblock{Morgan Kaufmann}.	
	
\bibitem{DaviesGrape}
	Robinson, S. P. and Davies, C. (2000)
	\newblock{Molecular biology of grape berry ripening}.
	\newblock{\em Australian Journal of Grape and Wine Research}, 
	\newblock{\bf 6},
	\newblock{175-188}.
	
	
\bibitem{Sachs}
	Sachs, K., Perez, O., Pe'er, D., Lauffenburger, D. A. and Nolan, G. P. (2005)
	\newblock{Causal protein-signaling networks derived from multiparameter single-cell data}.
	\newblock{\em Science}, 
	\newblock{\bf 308},
	\newblock{523-529}.

\bibitem{SGS}
	Spirtes, P. and Glymour, C. and Scheines, R. (1993)
	\newblock{\em Causation, Prediction, and Search}.
	\newblock{Springer-Verlag}.

\bibitem{Tricomi}
	Tricomi, F. G. and Erd\'elyi, A. (1951)
	\newblock{The asymptotic expansion of a ratio of gamma functions}.
	\newblock{\em Pacific Journal of Mathematics}, 
	\newblock{\bf 1},
	\newblock{133-142}.
	
\bibitem{MMHC}
	Tsamardinos, I. and Brown, L. E. and Aliferis, C. F. (2006)
	\newblock{The max-min hill-climbing {B}ayesian network structure learning algorithm}.
	\newblock{\em Machine Learning},
	\newblock{\bf 65},
	\newblock{31-76}.	
	
\bibitem{HSPs}
	Wang, W., Vinocur, B., Shoseyov, O. and Altman, A. (2004)
	\newblock{Role of plant heat-shock proteins and molecular chaperones in the abiotic stress response}.
	\newblock{\em Trends in Plant Science}, 
	\newblock{\bf 9},
	\newblock{244-252}.

\end{thebibliography}
\end{document}